\documentclass{ifacconf}

\usepackage{natbib}

\usepackage{graphicx}
\graphicspath{{Figures/}}
\usepackage{subcaption}

\usepackage{mathtools, amssymb, bm}

\newtheorem{proposition}{Proposition}
\theoremstyle{definition}

\usepackage{array}
\usepackage{enumerate}
\usepackage{enumitem}
\usepackage{xcolor}
\usepackage{balance}
\usepackage{multirow}
\usepackage{booktabs}

\allowdisplaybreaks

\makeatletter
\DeclareRobustCommand{\qed}{%
  \ifmmode 
  \else \leavevmode\unskip\penalty9999 \hbox{}\nobreak\hfill
  \fi
  \quad\hbox{\qedsymbol}}
\newcommand{\qedsymbol}{$\blacksquare$}
\newenvironment{proof}[1][\proofname]{\par
  \normalfont
  \topsep6\p@\@plus6\p@ \trivlist
  \item[\hskip\labelsep\itshape
    #1.]\ignorespaces
}{%
  \qed\endtrivlist
}
\newcommand{\proofname}{Proof}
\makeatother

\begin{document}
\begin{frontmatter}

\title{Analysis of the Unscented Transform Controller for Systems with Bounded Nonlinearities\thanksref{footnoteinfo}} 

\thanks[footnoteinfo]{This work was supported by Air Force Office of Scientific Research grant FA9550-23-1-0131 and NASA University Leadership Initiative grant 80NSSC22M0070. (Corresponding Author: Ram Padmanabhan. Email: \texttt{ramp3@illinois.edu}.)}

\author{Siddharth A. Dinkar$^1$,}
\author{Ram Padmanabhan$^1$,}
\author{Anna Clarke$^2$,}
\author{Per-Olof Gutman$^2$, and} 
\author{Melkior Ornik$^1$}

\address{$^1$University of Illinois Urbana-Champaign, Urbana, IL 61801, USA.}
\address{$^2$Technion -- Israel Institute of Technology, Haifa 32000, Israel.}

\begin{abstract}
In this paper, we present an analysis of the Unscented Transform Controller (UTC), a technique to control nonlinear systems motivated as a dual to the Unscented Kalman Filter (UKF). We consider linear, discrete-time systems augmented by a bounded nonlinear function of the state. For such systems, we review $1$-step and $N$-step versions of the UTC. Using a Lyapunov-based analysis, we prove that the states and inputs converge to a bounded ball around the origin, whose radius depends on the bound on the nonlinearity. Using examples of a fighter jet model and a quadcopter, we demonstrate that the UTC achieves satisfactory regulation and tracking performance on these nonlinear models.
\end{abstract}


\end{frontmatter}

\section{Introduction}
Controlling nonlinear systems presents unique challenges in comparison to linear control, especially in ensuring stability, achieving robustness to uncertainties, and managing computational complexity \citep{Khalil}. There exist a number of methods to control such systems, including nonlinear model-predictive control \citep{NMPC}, feedback linearization \citep{Khalil}, and more complex neural network-based controllers \citep{NN1, NN2}. Each of these methods suffer from drawbacks that limit their use in practical nonlinear control. Feedback linearization requires the use of an invertible coordinate transformation that is not guaranteed to exist for general nonlinear systems, and cannot handle systems with actuation constraints. Model predictive control and neural network-based methods require high computational effort, limiting their applicability in real-world problems.

The focus of this paper is on the Unscented Transform Controller (UTC), proposed by \cite{Clarke} and \cite{CG23}. The motivation behind this technique is a duality between optimal state estimation and control. {Such a duality is well-explored in the context of linear systems, where the Riccati equations for the Linear Quadratic Gaussian (LQG) controller and the Kalman filter have similar structures \citep[Chapter 7]{Astrom}.} \cite{Todorov08} describes a more general version of such a duality where the Kalman filter is replaced by the information filter. When the dynamics are nonlinear, state estimation is often accomplished through variants of the Kalman filter, including the Extended Kalman Filter \citep{JU04} and Unscented Kalman Filter (UKF) \citep{UKF, JU04}. 

The UKF relies on the Unscented Transform (UT), a technique to apply a nonlinear transformation to a given probability distribution \citep{UT}. This transform uses a set of sampled \emph{sigma points} to propagate statistics of the original distribution through the nonlinear transformation, rather than attempting to transform the distribution directly. In the UKF, this transform is used to obtain statistics of the estimated state after propagating its sigma points through nonlinear dynamics. 

Motivated by the notion of duality between state estimation and control, \cite{Clarke} and \cite{CG23} proposed the Unscented Transform Controller (UTC) for controlling nonlinear systems. In contrast to using the Unscented Transform to obtain a prediction of the state, the UTC uses this transform to `predict' a control input to achieve a desired trajectory \citep{Clarke}. Sigma points are chosen around the current control prediction rather than state estimate, and propagated through the nonlinear dynamics to obtain measurements corresponding to each sigma point of the control. These measurements are used to design an algorithm in a dual fashion to the UKF. \cite{Clarke} applied the UTC for controlling complex maneuvers during free fall of a vertical skydiver, and the dynamics of this system are nonlinear in nature. However, despite a discussion on stability for linear dynamics by \cite{Clarke}, there is no theoretical analysis of stability of the UTC under nonlinear dynamics. 

In this paper, we present an analysis of the UTC for discrete-time, time-invariant systems whose dynamics are described by a linear system augmented by nonlinear terms that are bounded in norm. We note that such systems are common in practice, and phenomena such as natural damping, passivity or geometric design may ensure that the linear portion of the dynamics are asymptotically stable \citep{Damping, Khalil}. {Our control objectives include both output regulation and trajectory tracking problems.} The analysis of stability uses a Lyapunov argument, and for the problem of output regulation, we show that the vector of states and inputs converges asymptotically to a bounded ball of a given radius depending on the bound on nonlinearities. We use two illustrative examples, on models of a fighter jet and a quadcopter, to show that the UTC achieves satisfactory performance in both output regulation and tracking.

The remainder of this paper is organized as follows. In Section \ref{sec:UTC}, we discuss the Unscented Transform Controller as introduced in previous work, providing the algorithm for its $1$-step and general $N$-step prediction variants. In Section \ref{sec:Stability}, we present a stability analysis of the $N$-step UTC algorithm, showing that the system states and inputs converge asymptotically to a bounded ball around the origin. We present two examples on a linear fighter-jet model and a nonlinear quadcopter model in Section \ref{sec:Examples}, illustrating the use of the UTC in these scenarios.

\section{Unscented Transform Controller} \label{sec:UTC}
Consider a time-invariant discrete-time control system of the form 
\begin{equation} \label{eq:System}
    x_{k+1} = \mathbf{A}x_{k} + \mathbf{B}u_{k} + f(x_{k}), ~~~y_k = \mathbf{C}x_k
\end{equation}
where $x_k \in \mathbb{R}^n$ is the state vector, $u_k \in \mathbb{R}^m$ is the control input, $y_k \in \mathbb{R}^p$ is the measurement vector and $f:\mathbb{R}^n \to \mathbb{R}^n$ is a bounded nonlinear function of the state $x_k$ satisfying
\begin{equation}
    \left\|f(x)\right\|_2 \leq \bar{f} \text{ for all $x \in \mathbb{R}^n$},
\end{equation}
for some $\bar{f} > 0$. {We assume the dynamics, including the nonlinear function $f$, are fully known.} The form of system \eqref{eq:System} is motivated by a linear system augmented with a bounded nonlinearity. {We consider the control objective of ensuring that the output $y_k$ tracks a reference trajectory $r_k$, with output regulation as a special case when $r_k = 0$ for all $k$.} We now provide details of the Unscented Transform Controller (UTC), proposed by \cite{Clarke} and \cite{CG23} to achieve this objective.

\subsection{1-Step UTC} \label{sec:1Step}
The UTC algorithm can be summarized into two steps: the \emph{propagation} step and the \emph{update} step. Based on the current state $x_k$ and input $u_k$, the \emph{propagation} step generates a control $u_{k+1}^{-}$. {This control may be based on a designer's prior knowledge which is used to specify a nonlinear feedback law
\begin{equation}
    u_{k+1}^- = h(x_k, u_k)
\end{equation}
for some function $h : \mathbb{R}^n \times \mathbb{R}^m \xrightarrow{} \mathbb{R}^m$. If no such prior knowledge is assumed, we suppose that the control is of the form 
\begin{equation}
    u_{k+1}^- = u_k + w_k; \quad w_k \sim \mathcal{N}(0, \mathbf{Q}_{u})
\end{equation}
for some initial controller covariance $\mathbf{Q}_{u}$.} The propagation step of the UTC then generates $2m + 1$ sigma points $U^i_k$ and their associated weights $W^i$, where $m$ denotes the dimension of $u_k \in \mathbb{R}^m$.
\begin{align}
    &U^0_k = u_k^-, \quad U^i_k = u_k^- \pm \sqrt{\frac{n}{1 - W^0}}S_j, \label{eq:u_sigma} \\
    &W^0 \in (0, 1), \quad W^i = \frac{1 - W^0}{2m}, \label{eq:w_sigma} \\ 
    \text{for } &i = 1, \ldots, 2m, \quad j = 1,\ldots, m. \nonumber
\end{align}
{Here, $S_j$ is the $j$th column of the square root of the matrix $\mathbf{P_{k|k}}$, which denotes the covariance matrix of the vector of control inputs $u_k$. In particular, at the beginning of each propagation step, we set $\mathbf{Q}_u = \mathbf{P_{0|0}}$. This covariance matrix is positive semi-definite, and hence its square root is unique and positive semi-definite \citep{Horn}.} Further, $W^0$ is a tuning parameter chosen by the designer. Each of the $2m+1$ sigma points are individually propagated through the dynamics \eqref{eq:System} to generate $2m+1$ outputs $Y^i$ as follows:
\begin{equation}
    Y^i = \mathbf{C}(\mathbf{A}x_{k} + \mathbf{B}U^i_k + f(x_{k})), ~~i = 1, \ldots, 2m,
\end{equation}
and $Y^0$ is generated using $U^{0}_{k}$ in \eqref{eq:u_sigma}. The time index $k$ is omitted in these outputs for simplicity of notation. As described by \cite{Clarke} and \cite{CG23}, define the \emph{predicted output} as the weighted sum of these points:
\begin{equation} \label{eq:ypred}
    y_{pred} = \sum_{i=0}^{2m}W^iY^i.
\end{equation}

The next step of the UTC is the \emph{update} step. We define the weighted mean of the sigma points by 
\begin{equation}
    u_{k+1|k} = \sum_0^{2m}W^iU^i.
\end{equation}
We then compute the predicted input covariance $\mathbf{P}_{k+1|k}$, the output covariance $\mathbf{P}_y$, and the input-output cross-covariance $\mathbf{P}_{uy}$ as follows:
\begin{subequations} \label{eq:P_update}
\begin{align}
    \mathbf{P}_{k+1|k} &= \mathbf{Q}_u + \sum_{i=0}^{2m}W^i  (U^i - u_{k+1|k})(U^i - u_{k+1|k})^T, \\
    \mathbf{P}_y &= \mathbf{P}_{err} + \sum_{i=0}^{2m}W^i (Y^i - y_{pred})(Y^i - y_{pred})^T, \\
    \mathbf{P}_{uy} &= \sum_{i=0}^{2m}W^i (U^i - u_{k+1|k})(Y^i - y_{pred})^T.
\end{align}
\end{subequations}
These matrices are used to compute a gain $\mathbf{K}$ to generate the desired control at the next time step as well as update the input covariance $\mathbf{P}_{k+1|k+1}$ \citep{Clarke, CG23}:
\begin{subequations} \label{eq:K_update}
\begin{align}
    \mathbf{K} &= \mathbf{P}_{uy}\mathbf{P}_y^{-1}, \\
    u_{k+1} &= u_{k+1|k} + \mathbf{K}(r_k - y_{pred}), \\
    \mathbf{P}_{k+1|k+1} &= \mathbf{P}_{k+1|k} - \mathbf{K}\mathbf{P}_y\mathbf{K}^T.
\end{align}
\end{subequations}
{These quantities are used to generate the $2m+1$ sigma points at the next time step based on \eqref{eq:u_sigma} and \eqref{eq:w_sigma}. It is clear that the update equations \eqref{eq:K_update} have a similar structure to those of the Unscented Kalman Filter (UKF) \citep{UKF}. In particular, the reference trajectory $r_k$ in the UTC plays the role of the true measurements in the UKF. Subsequently, the tracking error $r_k-y_{pred}$ in \eqref{eq:K_update} mimics the estimation error in the UKF. The UTC thus generates a control that attempts to ensure reference tracking by minimizing the tracking error. We remark here that input constraints, if present, may be addressed by imposing them on the sigma points at each prediction step \citep{Clarke, CG23}.}

\subsection{N-Step UTC} \label{sec:NStep}
{The UTC described above propagates the control through the dynamics over only $1$ time step. We now describe the general case where the control is propagated over $N$ steps, labeled the $N$-step UTC. In this formulation, each of the $2m+1$ input sigma point $U^i_k$ drives the model forward $N$ times:}
\begin{align}
    x^i_{k+1} &= \mathbf{A}x_{k} + \mathbf{B}U^i_k + f(x_{k}), \nonumber \\
    x^i_{k+2} &= \mathbf{A}x_{k+1} + \mathbf{B}U^i_{k} + f(x_{k+1}), \nonumber \\
    &\vdots \nonumber \\
    x^i_{k+N} &= \mathbf{A}x_{k+N-1} + \mathbf{B}U^i_{k} + f(x_{k+N-1}), ~~i = 0, \ldots, 2m, \label{eq:recur}
\end{align}
where $U^{i}_{k}$ are defined in \eqref{eq:u_sigma}. Due to the recursive relation in \eqref{eq:recur}, we can unroll $x^i_{k+N}$ in terms of $x_k$ and $U^i_k$. First, we define the matrices $\mathbf{F} = \mathbf{A}^N$, $\mathbf{G} = (\mathbf{A}^{N-1} + \mathbf{A}^{N-2} + ~\cdots ~+ \mathbf{A} + \mathbf{I}_n)$, and $g(x_k, N) = \sum_{i=0}^{N-1}\mathbf{A}^{N-i-1}f(x_{k+i})$ to simplify the $N$-step propagation:
\begin{equation} \label{eq:System_Nstep}
    x^i_{k+N} = \mathbf{F}x_k + \mathbf{GB}U_{k}^{i} + g(x_k, N).
\end{equation}
{Propagating each sigma point $U_{k}^{i}$ through the dynamics, we generate $2m+1$ outputs $Y^i = \mathbf{C}x^i_{k+N}$. The algorithm then follows the steps starting from \eqref{eq:ypred}, generating the control at the next instant from \eqref{eq:K_update}. We note that the above algorithm generalizes the $1$-step case. In particular, by setting $N = 1$, the terms $\mathbf{F}$, $\mathbf{G}$, and $g(x_k,N)$ simply reduce to $\mathbf{A}$, $\mathbf{I}_n$, $f(x_k)$ respectively, while the control update reduces to the standard $1$‐step update.}

{In the following section, we consider the case when $r_k = 0$ for all $k$, i.e., the problem of output regulation. In this setting, we analyze the stability of the closed-loop system in terms of the norm of states and inputs.}

\section{Stability Analysis} \label{sec:Stability}
{In this section, we consider the output regulation problem and present a Lyapunov-based analysis of the $N$-step UTC. \cite{Clarke} provides a compact description for the $N$-step UTC applied to a linear system. Using \eqref{eq:System_Nstep}, we can similarly write the following description of the system dynamics and controller:}
\begin{align} \label{eq:closed-loop}
\begin{split}
    x_{k+1} &= \mathbf{A}x_k + \mathbf{B}u_k + f(x_k), \\
    y_{k} &= \mathbf{C}x_k, \\
    u_k &= u_{k-1} - \mathbf{K}\mathbf{C}(\mathbf{F}x_k + \mathbf{GB}u_{k-1} + g(x_k, N)).
\end{split}
\end{align}
Let $v_{k+1} = u_k$, so that we can write the following augmented dynamics:
\begin{align*}
    \begin{bmatrix}
    x_{k+1} \\
    v_{k+1}
  \end{bmatrix} &= 
  \begin{bmatrix}
    (\mathbf{A} - \mathbf{B}\mathbf{K}\mathbf{C}\mathbf{F}) & (\mathbf{I}-\mathbf{B}\mathbf{K}\mathbf{C}\mathbf{G})\mathbf{B}\\
    -\mathbf{K}\mathbf{C}\mathbf{F} & \mathbf{I}-\mathbf{K}\mathbf{C}\mathbf{G}\mathbf{B}
  \end{bmatrix}
\begin{bmatrix}
    x_{k} \\
    v_{k}
  \end{bmatrix} \nonumber \\ 
  &+ 
  \begin{bmatrix}
    f(x_k)-\mathbf{B}\mathbf{K}\mathbf{C}g(x_k, N) \\
    -\mathbf{K}\mathbf{C}g(x_k, N)
  \end{bmatrix}.
\end{align*}
To simplify notation further, let
\begin{gather*}
a_{k} = \begin{bmatrix}
    x_{k} \\
    v_{k}
  \end{bmatrix}; \quad 
  \mathbf{Z} = \begin{bmatrix}
    (\mathbf{A} - \mathbf{B}\mathbf{K}\mathbf{C}\mathbf{F}) & (\mathbf{I}-\mathbf{B}\mathbf{K}\mathbf{C}\mathbf{G})\mathbf{B}\\
    -\mathbf{K}\mathbf{C}\mathbf{F} & \mathbf{I}-\mathbf{K}\mathbf{C}\mathbf{G}\mathbf{B}
  \end{bmatrix}; \\ 
  D_k = \begin{bmatrix}
    f(x_k)-\mathbf{B}\mathbf{K}\mathbf{C}g(x_k, N) \\
    -\mathbf{K}\mathbf{C}g(x_k, N)
  \end{bmatrix},
\end{gather*}
so that the system can be compactly written as
\begin{equation} \label{eq:lyap_system}
    a_{k+1} = \mathbf{Z}a_k + D_k,
\end{equation}
where $D_k$ captures the nonlinearity of the system. We note that the gain $\mathbf{K}$ can then be chosen to ensure that the linear system $a_{k+1} = \mathbf{Z}a_k$ is asymptotically stable. Thus, there exists a positive definite matrix $\mathbf{P}$ such that $\mathbf{Z}^T\mathbf{P}\mathbf{Z} - \mathbf{P} = -\mathbf{I}$ \citep[Chapter 3]{Khalil}. We now use \eqref{eq:lyap_system} to analyze the stability of the closed-loop description \eqref{eq:closed-loop} in the following proposition.

\begin{proposition}
Under the system dynamics and controller in \eqref{eq:closed-loop} which describe the $N$-step UTC applied to a linear system affected by a nonlinearity with norm bound $\bar{f}$, the vector of states and inputs $a_k$ with the dynamics \eqref{eq:lyap_system} converges to a bounded ball of radius $R$, where
\begin{align}
    R &= \Bar{D}\left[\|\mathbf{Z}\|p_{max} + \sqrt{{\|\mathbf{Z}\|}^2{p_{max}}^2 + p_{max}}\right], \\
    \Bar{D} &= \sqrt{(\Bar{f} + \|\mathbf{B}\mathbf{K}\mathbf{C}\|_2 \Bar{g})^2 + (\|\mathbf{K}\mathbf{C}\|_2 \Bar{g})^2}, \nonumber \\
    \Bar{g} &= \Bar{f}\sum_{i=0}^{N-1}\|\mathbf{A}^{N-i-1}\|_2, \nonumber
\end{align}
and $p_{max}$ is the maximum eigenvalue of the matrix $\mathbf{P}$ satisfying $\mathbf{Z}^T\mathbf{P}\mathbf{Z} - \mathbf{P} = -\mathbf{I}$.
\end{proposition}
\begin{proof}
Define a Lyapunov function
\begin{equation}
    V_k := {a_k}^T\mathbf{P}a_k
\end{equation}
and $\Delta V_k = V_{k+1} - V_k$. Our aim is to investigate conditions under which $\Delta V_k <  0$, to ensure asymptotic stability is achieved \citep{Khalil}. Substituting into the equation and simplifying according to our assumption of the linear system's stability, we get
\begin{align}
    \Delta V_k &= [\mathbf{Z}a_k + D_k]^T\mathbf{P}[\mathbf{Z}a_k + D_k] - {a_k}^T\mathbf{P}a_k \nonumber \\
    &= {a_k}^T(\mathbf{Z}^T\mathbf{P}\mathbf{Z} - \mathbf{P})a_k + 2{a_k}^T\mathbf{Z}^T\mathbf{P}D_k + {D_k}^T\mathbf{P}D_k \nonumber \\
    &= -\|a_k\|_2^2 + 2{a_k}^T\mathbf{Z}^T\mathbf{P}D_k + {D_k}^T\mathbf{P}D_k. \label{eq:DeltaVk1}
\end{align}
Since $\mathbf{P}$ is positive definite, the Rayleigh inequality $p_{min}\|D_k\|_2^2 \leq {D_k}^T\mathbf{P}D_k \leq p_{max}\|D_k\|_2^2$ holds, where $p_{min}, p_{max}$ are the minimum and maximum eigenvalues of $\mathbf{P}$ respectively \citep{Horn}. We use the upper bound $p_{max}\|D_k\|_2^2$ for our purposes. By the Cauchy-Schwartz inequality, we can simplify
\begin{align}
    {a_k}^T\mathbf{Z}^T\mathbf{P}D_k &= [\mathbf{Z} a_k]^T\mathbf{P}D_k \leq \|\mathbf{Z}a_k\|_2 \|\mathbf{P}D_k\|_2 \nonumber \\
    &\leq \|\mathbf{Z}\| \|\mathbf{P}\| \|D_k\|_2 \|a_k\|_2 \nonumber \\
    &\leq p_{max} \|\mathbf{Z}\| \|D_k\|_2 \|a_k\|_2. \label{eq:ak1}
\end{align}
Substituting \eqref{eq:ak1} and the preceding inequality in \eqref{eq:DeltaVk1},
\begin{equation} \label{eq:DeltaVk2}
    \Delta V_k \leq \\-\|a_k\|_2^2 + 2\|\mathbf{Z}\| p_{max} \|D_k\|_2 \|a_k\|_2 + p_{max}\|D_k\|_2^2.
\end{equation}
Asymptotic stability holds when $\Delta V_k < 0$, and using \eqref{eq:DeltaVk2}, this condition is satisfied when
\begin{equation}
\begin{split}
    \|a_k\|_2^2
    - 2\|\mathbf{Z}\| p_{max} \|D_k\|_2 \|a_k\|_2
    - p_{max}\|D_k\|_2^2 > 0
\end{split}
\end{equation}
Note that this inequality is quadratic in $\|a_k\|_2$. Solving for $\|a_k\|_2$, the condition $\Delta V_k < 0$ holds if 
\begin{equation} \label{eq:IBOUND}
    \|a_k\|_2 > \|D_k\|_2\left[\|\mathbf{Z}\|p_{max} + \sqrt{{\|\mathbf{Z}\|}^2{p_{max}}^2 + p_{max}}\right].
\end{equation}
However, $\|D_k\|_2$ varies with the timestep $k$. Recall that $D_k = \begin{bmatrix}
    f(x_k)-\mathbf{B}\mathbf{K}\mathbf{C}g(x_k, N) \\
    -\mathbf{K}\mathbf{C}g(x_k, N)
  \end{bmatrix}$ and $\|f(x_k)\|_2 \leq \Bar{f}$.
Since $g(x_k, N) = \sum_{i=0}^{N-1}\mathbf{A}^{N-i-1}f(x_{k+i})$, by the triangle inequality,
\begin{align}
\begin{split}
    \|g(x_k, N)\|_2 &= \left\|\sum_{i=0}^{N-1}\mathbf{A}^{N-i-1}f(x_{k+i}) \right\|_2 \\
    &\leq \sum_{i=0}^{N-1}\left\|\mathbf{A}^{N-i-1}f(x_{k+i})\right\|_2 \\
    &\leq \Bar{f}\sum_{i=0}^{N-1}\|\mathbf{A}^{N-i-1}\|_2 
    =\vcentcolon \Bar{g}.
\end{split}
\end{align}
To derive a tight bound on the nonlinearity $D_k$ in terms of the bounds $\Bar{f}$ and $\Bar{g}$, note that
\begin{equation}
    \|D_k\|_2 = \sqrt{\|f(x_k)-\mathbf{B}\mathbf{K}\mathbf{C}g(x_k, N)\|_2^2 + \|\mathbf{K}\mathbf{C}g(x_k, N)\|_2^2}.
\end{equation}
We can apply the triangle inequality inside each of the norms in $\|D_k\|_2$ and use the submultiplicative property of matrix norms to get the bounds
\begin{align*}
        \|f(x_k)-\mathbf{B}\mathbf{K}\mathbf{C}g(x_k, N)\|_2 &\leq \Bar{f} + \|\mathbf{B}\mathbf{K}\mathbf{C}\|_2 \Bar{g}, \\
        \|\mathbf{K}\mathbf{C}g(x_k, N)\|_2 &\leq \|\mathbf{K}\mathbf{C}\|_2 \Bar{g}.
\end{align*}
It follows that 
  \begin{equation}
  \begin{split}
      \|D_k\|_2 \leq \sqrt{(\Bar{f} + \|\mathbf{B}\mathbf{K}\mathbf{C}\|_2 \Bar{g})^2 + (\|\mathbf{K}\mathbf{C}\|_2 \Bar{g})^2} 
      =\vcentcolon \Bar{D}.
    \end{split}
  \end{equation}
Thus, if
\begin{equation}
  \|a_k\|_2 > \Bar{D}\left[\|\mathbf{Z}\|p_{max} + \sqrt{{\|\mathbf{Z}\|}^2{p_{max}}^2 + p_{max}}\right],
\end{equation}
then \eqref{eq:IBOUND} is satisfied and $\Delta V_k < 0$. Thus, the state $a_k$ converges asymptotically to a bounded ball with radius $R$, where 
$$
    R = \Bar{D}\left[\|\mathbf{Z}\|p_{max} + \sqrt{{\|\mathbf{Z}\|}^2{p_{max}}^2 + p_{max}}\right],
$$
thus concluding the proof.
\end{proof}
We note that the analysis above is potentially conservative due to the use of different inequalities. Improving this is a useful avenue for future work. We also emphasize that this stability guarantee depends on the number of prediction steps $N$, as the bound $\Bar{D}$ and radius $R$ clearly depend on $N$.

\section{Simulation Examples} \label{sec:Examples}
In this section, we illustrate the efficacy and flexibility of the UTC on two representative systems. Section \ref{sec:ADMIRE} focuses on attitude control of the ADMIRE fighter jet model \citep{SDRA}. We augment the model with various known nonlinear functions $f(x_k)$, which represent possible environmental disturbances. This model is thus a useful test case for the UTC based on the system formulation \eqref{eq:System}. Section \ref{sec:Quad} then investigates the performance of the UTC on a nonlinear quadcopter system, which serves as a canonical example in robotics and unmanned aerial vehicles research. In both settings, we demonstrate how the UTC’s design achieves satisfactory stabilization and tracking performance.

\begin{figure}[!t]
\centering
\includegraphics[width = 0.45\textwidth]{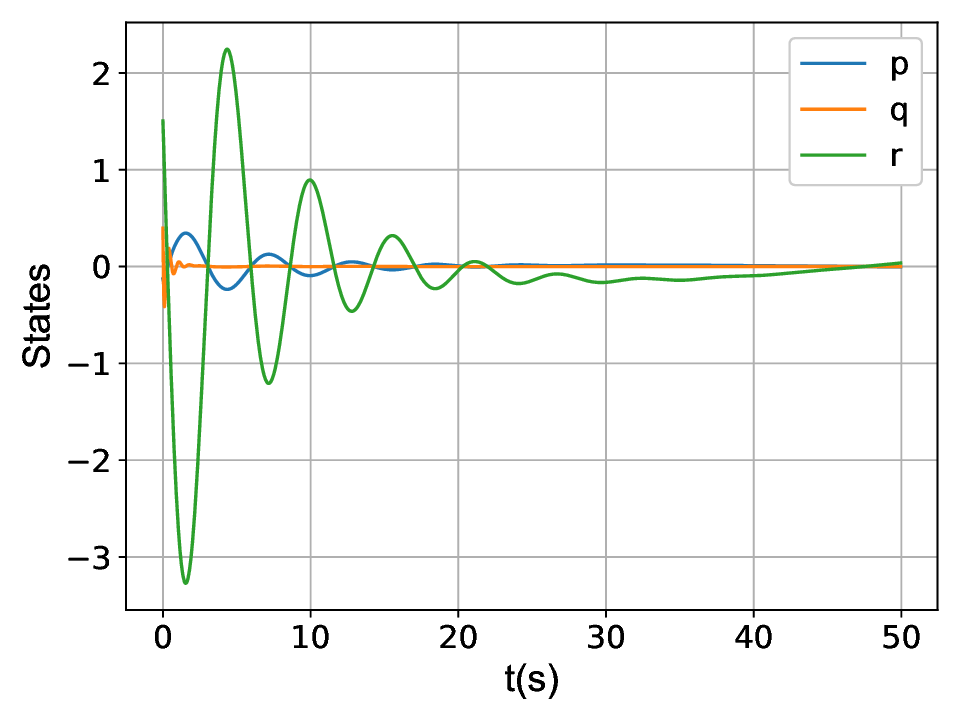}
\caption{Performance of the $1$-step UTC in stabilizing states of the ADMIRE model to the origin.}
\label{figure:ADMIRE_reg}
\end{figure}

\subsection{Example 1: ADMIRE Fighter Jet Model} \label{sec:ADMIRE}
We use the dynamics of the ADMIRE fighter-jet model \citep{SDRA}, a popular application for control frameworks \citep{KWT10, BO22}. We focus on the attitude control subsystem of the model, which consists of states $p,q$ and $r$ denoting the roll, pitch and yaw rates in rad/s. The linearized dynamics of this subsystem were established by \cite{SDRA} as $\dot x = \mathbf{A}x(t) + \mathbf{B}u(t)$ where, 
\begin{align}
\begin{split} \label{eq:ADMIRE}
    x &= \begin{bmatrix}p \\ q \\ r\end{bmatrix}; \quad \mathbf{A} = \begin{bmatrix}
    -0.9967  & 0 & 0.6176 \\ 
    0 &-0.5057 &0 \\
    -0.0939 &0 &-0.2127
    \end{bmatrix}; \\
    \mathbf{B} &= \begin{bmatrix}
    0 & -4.2423 & 4.2423 & 1.4871 \\ 
    1.6532 & -1.2735 & -1.2735 & 0.0024 \\
    0 & -0.2805 & 0.2805 & -0.8823
    \end{bmatrix};  \quad
    u = \begin{bmatrix} u_c \\ u_{re} \\ u_{le} \\ u_r \end{bmatrix}.
\end{split}
\end{align}
The control input $u$ contains the deflection angle in radians of $4$ actuators, where $u_c$ is the canard wing deflection, $u_{re}$ and $u_{le}$ are the right and left elevon deflections respectively, and $u_r$ is the rudder deflection. We consider the case with no input constraints here, but note that such constraints may be handled by imposing them at each prediction step. We assume that the measurements consist of the entire state $x_k$, and the objective is to stabilize these measurements so that each component of the state converges to zero. The linear subsystem \eqref{eq:ADMIRE} is augmented with nonlinear, sinusoidal perturbations of the form $\beta \sin(\alpha x_k)$ added to the dynamics of each component of the state $x_k$. These perturbations are clearly bounded by some $\Bar{f}$ that depends on the different values of $\beta$ in each component.

Fig.~\ref{figure:ADMIRE_reg} shows the performance of the $1$-step UTC in stabilizing the state to the origin. We choose a particular instance of randomly chosen initial conditions and sinusoidal perturbations. We notice that the sinusoidal perturbations introduce large transient errors which are eventually stabilized by the UTC. In particular, each component of the state converges to zero as desired. The performance in Fig.~\ref{figure:ADMIRE_reg} improves upon the quantitative bound shown in Section~\ref{sec:Stability}, where we only guarantee that the vector of states and inputs converges to a bounded ball around the origin. This behavior is because the analysis in Section~\ref{sec:Stability} is conservative, leading to potentially improved performance in practice.


\subsection{Example 2: Quadcopter} \label{sec:Quad}
We now illustrate the capabilities of the UTC on the dynamical model of a quadcopter \citep{Quadcopter}. 
Similar to the ADMIRE model, we focus on the attitude control problem. With the quadcopter model, we consider two separate frames of reference: the Earth's inertial frame containing Euler angles $z = \begin{bmatrix}\phi \quad \theta \quad \psi \end{bmatrix}^T$ corresponding to roll, pitch and yaw, and the quadcopter's body frame containing the body angular velocities $v = \begin{bmatrix} p \quad  q \quad r\end{bmatrix}^T$. We then have the dynamics
\begin{equation} \label{eq:z}
\dot{z} = \mathbf{R}(\phi, \theta, \psi)v,
\end{equation}
where
$$
\mathbf{R}(\phi, \theta, \psi) = \begin{bmatrix}
    1 & \sin{\phi}\tan{\theta} & \cos{\phi}\tan{\theta} \\
    0 & \cos{\phi} & -\sin{\phi} \\
    0 & \frac{\sin{\phi}}{\cos{\theta}} & \frac{\cos{\phi}}{\cos{\theta}}
\end{bmatrix}
$$
represents the rotation matrix from the body frame for $v$ to the inertial frame for $z$. Further, the dynamics in the body frame are 
\begin{equation} \label{eq:quad}
   \begin{bmatrix}\dot p \\ \dot q \\ \dot r\end{bmatrix} =  
   \begin{bmatrix} \frac{(I_{yy} - I_{zz})qr}{I_{xx}} \\ \frac{(I_{zz} - I_{xx})pr}{I_{yy}} \\ \frac{(I_{xx} - I_{yy})pq}{I_{zz}}\end{bmatrix} - J_r \begin{bmatrix}\frac{q}{I_{xx}} \\ \frac{-p}{I_{yy}} \\ 0\end{bmatrix}\omega_r + \begin{bmatrix}\frac{\tau_\phi}{I_{xx}} \\ \frac{\tau_\theta}{I_{yy}} \\  \frac{\tau_\psi}{I_{zz}} \end{bmatrix}
\end{equation}
where $\omega_r = -\omega_1 + \omega_2 - \omega_3+ \omega_4$. We note that the bilinear terms $qr$, $pr$ and $pq$ in \eqref{eq:quad} contribute to the nonlinear behavior of the system. The body torques $\tau_B$ are derived from the control input $u$ as follows
\begin{equation}
 \tau_B = \begin{bmatrix}\tau_\phi \\ \tau_\theta \\  \tau_\psi \end{bmatrix} = 
 \begin{bmatrix}Lk_t(-\omega_2^2 + \omega^2_4) \\ Lk_t(-\omega_1^2 + \omega^2_3) \\ k_b(-\omega_1^2 + \omega_2^2 - \omega_3^2 + \omega_4^2) \end{bmatrix}; ~~ u = \begin{bmatrix}\omega_1 \\ \omega_2 \\ \omega_3 \\ \omega_4 \end{bmatrix}.
\end{equation}

\begin{figure}[!t]
\centering
\includegraphics[width = 0.45\textwidth]{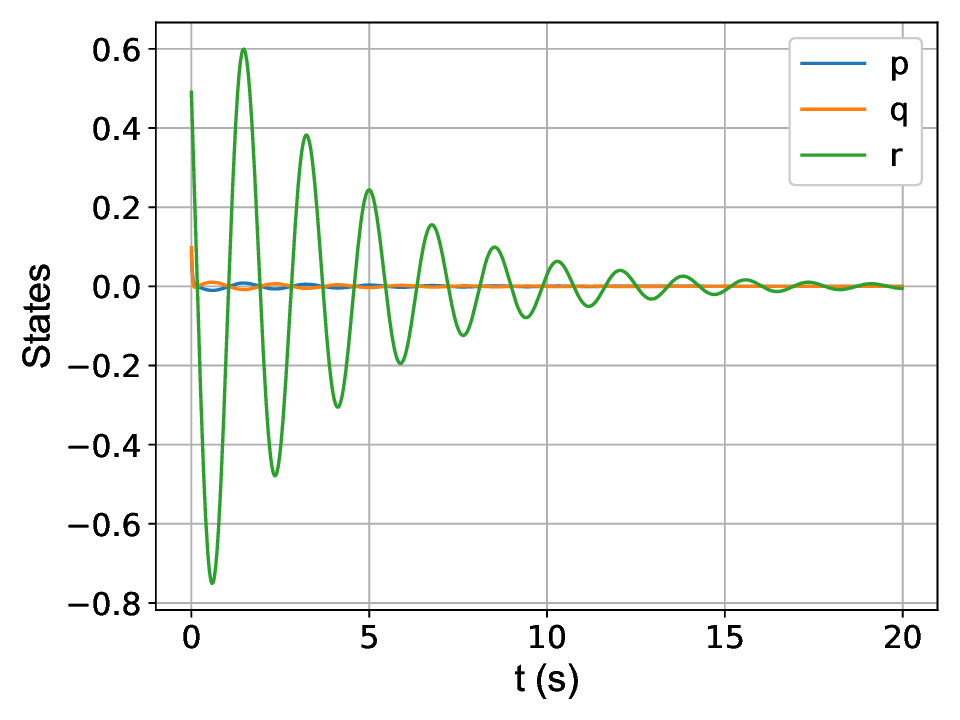}
\caption{Output regulation in a quadcopter model using a $3$-step UTC.}
\label{figure:quad_reg_3}
\end{figure}

\begin{figure}[!t]
\centering
\includegraphics[width = 0.45\textwidth]{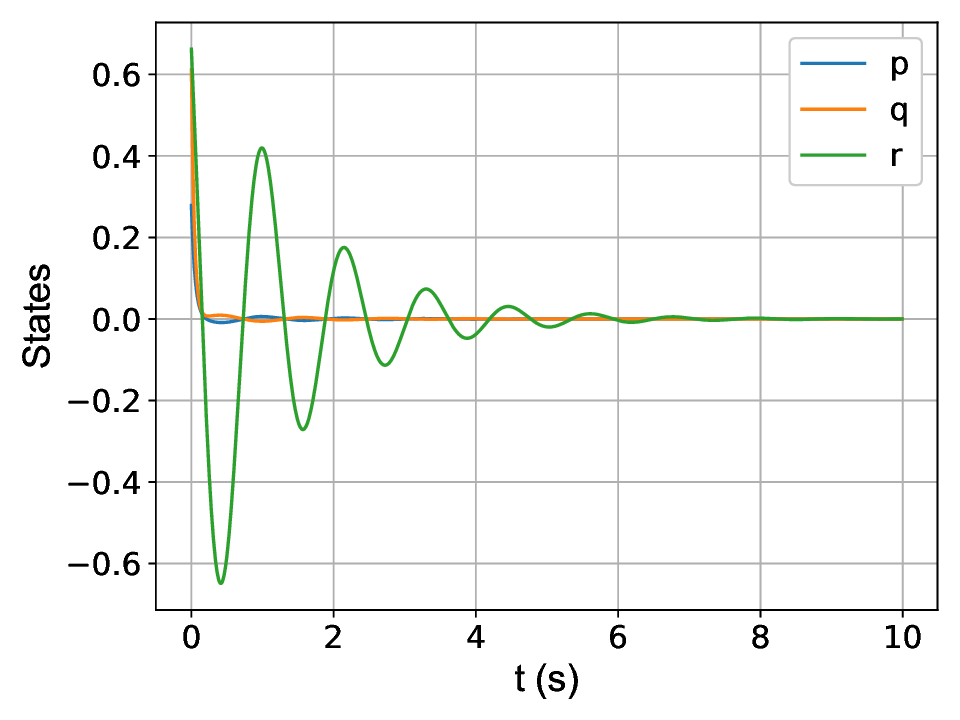}
\caption{Output regulation in a quadcopter model using a $5$-step UTC.}
\label{figure:quad_reg_5}
\end{figure}

\begin{figure}[!t]
\centering
\includegraphics[width = 0.45\textwidth]{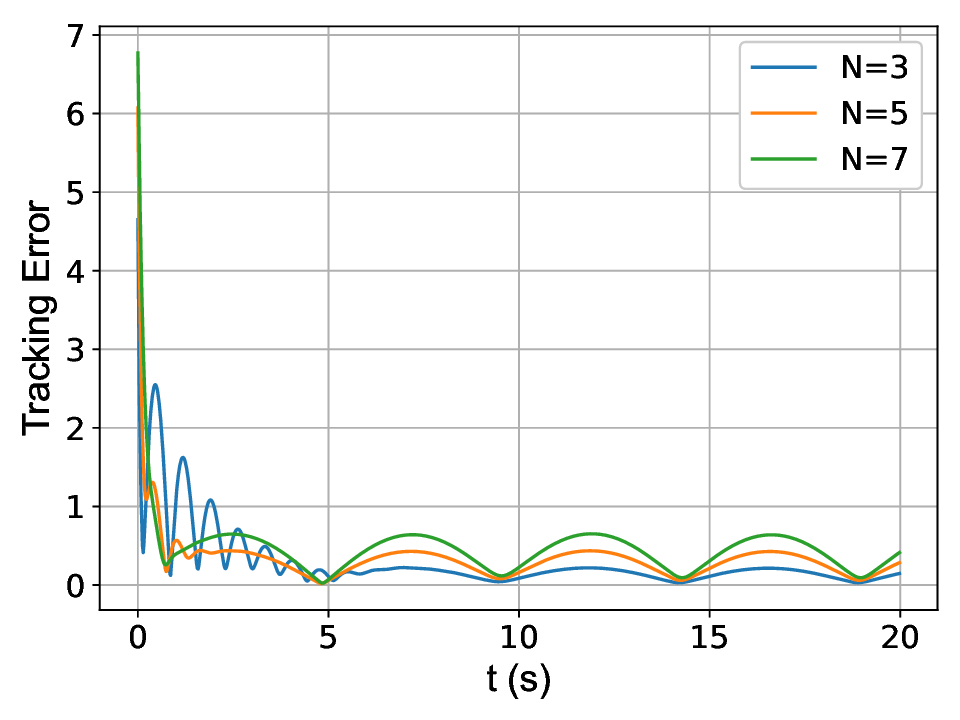}
\caption{Tracking performance of the UTC applied to a nonlinear quadcopter model, in terms of the norm of the tracking error over time.}
\label{figure:err_norm}
\end{figure}

The control input $u$ represents the rotor speed of each rotor of the quadcopter. The state vector consists of all Euler angles and body angular velocities, i.e., $x = [z^T ~~v^T]^T$ and the outputs correspond to the angular velocities, $y = v$. As before, we do not assume input constraints in this example. The constants $I_{xx}$, $I_{yy}$, and $I_{zz}$ represent the moment of inertia of the structure along the labeled axes, while $J_r$, $k_t$, $k_b$, and $L$ are constants related to the dimensions and thrust coefficients of the dynamic model in \cite{Quadcopter}. We note that while there are linear components to these dynamics arising from \eqref{eq:z}, the system is nonlinear in nature and is not strictly based on dynamics of the form \eqref{eq:System}. However, this example serves to illustrate the capabilities of the UTC for reference tracking even for general nonlinear dynamics which may not satisfy the condition of bounded nonlinearities.


In Figs.~\ref{figure:quad_reg_3} and \ref{figure:quad_reg_5}, we illustrate the performance of the UTC in stabilizing the angular velocities $v$ from randomly chosen initial conditions, using a $3$-step and $5$-step UTC algorithm. We note that the yaw rate $r$ displays overshoot due to the weak authority of the rotors over the body frame angular velocity. An important feature revealed in these figures is the benefit of the $N$-step prediction, which mitigates overshoot and shortens settling time. In particular, the $5$-step prediction reduces transients and leads to faster stabilization of the outputs. When sigma points are propagated through dynamics for a smaller number of steps, the resulting outputs have lower variance. This similarity between sigma points results in a large gain $\mathbf{K}$ as the output covariance matrix $\mathbf{P}_y$ is inverted in \eqref{eq:K_update}. Such a large gain can result in the worse transient behavior and slower settling time shown in Fig.~\ref{figure:quad_reg_3} for smaller number of prediction steps. On the other hand, a larger number of prediction steps results in higher computational costs in generating and propagating these sigma points.

In Fig.~\ref{figure:err_norm}, we compare the performance of different prediction sizes in the UTC for tracking sinusoidal references. Due to the highly nonlinear nature of the dynamics, exact tracking is more difficult to achieve; however, the errors converge to a ball around the origin. Interestingly, we note that the transients are attenuated quicker when the number of prediction steps is larger. As the number of prediction steps increases, the controller is able to respond aggressively to immediate deviations, effectively damping the initial overshoot and subsequent oscillations in short order. However, once the transient response subsides, the steady‐state error tends to be more pronounced for a larger number of prediction steps, owing to the nonlinear nature of the plant and the relatively high frequency of the sinusoidal references we choose. Increasing the step size cannot compensate for higher‐order or more rapid nonlinearities over the chosen prediction horizon, leading to the observed residual error. This effect is also evident with higher frequency reference trajectories, where the large prediction horizon conceals high frequency behavior leading to larger tracking errors. There is thus a tradeoff between transient performance and tracking accuracy as the number of prediction steps increases. 

\section{Conclusions} \label{sec:Conclusions}
In this paper, we have presented an analysis of the recently introduced Unscented Transform Controller (UTC) for systems with bounded nonlinearities. The UTC behaves as a dual to the Unscented Kalman Filter (UKF) for nonlinear filtering, and is motivated by the known duality between optimal estimation and control. We reviewed the UTC algorithm for such systems, providing versions with $1$ and $N$ prediction steps. Our analysis of stability used Lyapunov theory, proving that the states and inputs converge to a bounded ball whose radius is dependent on the norm bound on the nonlinearity. We presented two examples, on a fighter jet model and a model of a quadcopter, to demonstrate that the UTC achieves satisfactory regulation and tracking performance. In particular, we demonstrate a trade-off between transient performance and accuracy of tracking as the number of prediction steps increases.

Avenues for future work involve analysis of general nonlinear systems and continuous-time systems, possibly based on tools developed for the UKF by \cite{Sarkka07}. In this context, improving the conservatism of the analysis in Section \ref{sec:Stability} is an important task. Further examining the tradeoff between transient performance and tracking accuracy as discussed in Section \ref{sec:Quad} is another useful idea to explore. At a higher level, the design of the UTC and the motivation behind it spawn questions about the possibility of designing dual controllers for a variety of other estimation algorithms, such as the Extended Kalman Filter and the particle filter.


\balance

\bibliography{references}

\end{document}